\documentclass[11pt]{article}
\usepackage{mathrsfs}
\usepackage{CJK}
\usepackage{bm}
\usepackage{graphicx}
\usepackage{amsmath,amssymb}
\usepackage{amsfonts}
\usepackage{amsthm}
\usepackage{verbatim}
\usepackage{color}
\usepackage{graphicx}
\usepackage{amssymb}
\usepackage{latexsym}
\allowdisplaybreaks
\large\normalsize

\newcommand{\be}{\begin{equation}}
\newcommand{\ee}{\end{equation}}
\newcommand{\bes}{\begin{equation}\begin{aligned}}
\newcommand{\ees}{\end{aligned}\end{equation}}
\newcommand{\ben}{\begin{equation}\nonumber\begin{aligned}}

 \newcommand{\R}{\mathbb{R}}

\newcommand{\N}{\mathbb{N}}

\renewcommand{\leq}{\leqslant}
\renewcommand{\geq}{\geqslant}

\newtheorem{thm}{Theorem}[section]
\newtheorem{lem}[thm]{Lemma}

\newtheorem*{main thm}{Main Theorem}

\numberwithin{equation}{section}
\begin{document}

\title{A sharp recovery condition for  block sparse signals by  block orthogonal multi-matching  pursuit}

\author{Wengu~Chen and Huanmin~Ge
\thanks{W. Chen is with Institute of Applied Physics and Computational Mathematics,
Beijing, 100088, China, e-mail: chenwg@iapcm.ac.cn.}
\thanks{H. Ge is with Graduate School, China Academy of Engineering Physics,
Beijing, 100088, China, e-mail:gehuanmin@163.com.}
\thanks{This work was supported by the NSF of China (Nos.11271050, 11371183)
.} }

\maketitle

\begin{abstract}
We consider the block orthogonal multi-matching  pursuit (BOMMP) algorithm for the recovery of block sparse signals.
A sharp bound is obtained  for the exact reconstruction of  block $K$-sparse signals via the BOMMP algorithm in the noiseless case, based on the block restricted isometry constant (block-RIC). Moreover,
 we show that the sharp bound  combining with an extra condition on the minimum $\ell_2$ norm of
 nonzero blocks of block $K-$sparse signals  is sufficient to recover the true support of block $K$-sparse
  signals by the BOMMP  in the noise case. The significance of the results we obtain in this paper lies
   in the fact that making explicit use of block sparsity of block sparse signals can achieve  better recovery
    performance than ignoring the additional structure in the problem as being in the conventional sense.
\end{abstract}


\begin{bfseries}
Mathematics Subject Classification (2010)
\end{bfseries}
65D15, 65J22, 68W40

\section{Introduction}

 The framework of compressed sensing is concerned with the reconstruction of unknown sparse signals from an underdetermined linear system in \cite{h1},\cite{h2}. More concretely, this can be described as
 \begin{eqnarray}\label{pro1}
y=Ax+e,
\end{eqnarray}
where $y\in \R^m$ is a vector of measurements, the matrix $A\in\R^{m\times n}$ with $m\ll n$ is a known sensing matrix, the vector $x\in\R^n$ is a unknown $K$-sparse signal $(K\ll n)$ and $e\in\R^m$ is measurement error.  The goal is to recover the unknown signal $x$ based on $y$ and $A$. It has triggered  different  efficient methods which can be proved to recover  unknown $K$-sparse signals $x$ under a variety of different conditions on sensing matrix $A$ \cite{2}-\cite{12'}.

In this paper, we consider the unknown signal $x$ of the model \eqref{pro1} that exhibits additional structure
 in the form of the nonzero coefficients occurring in blocks. Such signal is called block sparse signal \cite{g13}, \cite{c1}.  We explicitly take this block structure into account to recover block signals
through the BOMMP algorithm.
Block sparse signals arise naturally  in many fields including DNA microarrays \cite{g10}, equalization of sparse communication \cite{g11},
multi-band signals \cite{g12}-\cite{h3} and the multiple measurement vector (MMV) problem \cite{c2}-\cite{c6}.

 Following  \cite{g13}, \cite{g14}, a block sparse signal $x\in\R^n$ over $\mathcal{I}=\{d_1,d_2,\ldots,d_l\}$ is a concatenation of $l$ blocks of length $d_i\ (i=1,2,\cdots,l)$, i.e.,
\begin{eqnarray}\label{m1}
x=[\underbrace{x_1 \ldots\ x_{d_1}}_{x'[1]} \underbrace{x_{d_1+1}\ldots\  x_{d_1+d_2}}_{x'[2]}\ldots \underbrace{x_{n-d_l+1}\ldots x_n}_{x'[l]}]^{'}
\end{eqnarray}
where $x[i]$ denotes the $i$th block of $x$ and $n=\sum_{i=1}^ld_i$. $x$ is called block $K-$sparse if $x[i]$ has nonzero $\ell_2$ norm for at most $K$ indices $i$. That is, $\sum_{i=1}^lI(\|x[i]\|_2>0)\leq K$, where $I(\cdot)$ is an indicator function. Denote $\|x\|_{2,0}=\sum_{i=1}^lI(\|x[i]\|_2>0)$ or $T=\textmd{block-supp}(x)=\{i:\|x[i]\|_2>0, \ i=1,2,\cdots,l\}$,
then a block $K-$sparse signal $x$ satisfies  $\|x\|_{2,0}\leq K$ and $|T|\leq K$. If $d_i=1\ (i=1,2,\cdots,l)$, the block sparse signal reduces to the conventional sparse signal \cite{h1}, \cite{h2}.
Similar to \eqref{m1},
sensing matrix $A$ can be expressed
as a concatenation of $l$ column blocks, i.e.,
\begin{eqnarray*}
 A=[\underbrace{A_1\ldots A_{d_1}}_{A[1]} \underbrace{A_{d_1+1}\ldots A_{d_1+d_2}}_{A[2]} \ldots \underbrace{A_{n-d_l+1}\ldots A_n}_{A[l]}],
\end{eqnarray*}
where $A_i$ is the $i$th column of $A$ for $i=1,2,\cdots,n$.

 To recover block
sparse signals $x$, one approach to exploiting block sparsity is  the mixed $\ell_2/\ell_0$ norm minimization:
\begin{eqnarray*}
  \min_{x}\|x\|_{2,0}\ \ \textmd{subject} \ \textmd{to} \ \ \|Ax-y\|_2\leq\varepsilon,
\end{eqnarray*}
where $\varepsilon$ is the noise level. In noiseless case,
$\varepsilon=0$. The  minimization problem is a suitable extension
of the standard $\ell_0$-minimization problem.
 This minimization problem is also NP-hard. Instead, some efficient
methods making explicit use of block sparsity to imply the recovery of block sparse signals include  the mixed $\ell_2/\ell_1$  norm minimization \cite{g13}, \cite{g14}, \cite{g13''}-\cite{cc6}, the mixed $\ell_2/\ell_p(0<p<1)$ norm minimization \cite{g15}-\cite{g17}, the BOMP algorithm \cite{g14},\cite{g17',g17'',g18}, the sparsity adaptive regularized  OMP algorithm  \cite{g18'},
the block version of StOMP algorithm \cite{g18''}.

To investigate the recovery of block sparse signals, Eldar and Mishali introduced the notion of the block restricted isometry property(block-RIP) and also demonstrated  that
  the block-RIP has advantages over standard RIP  in \cite{g13}.
Sensing matrix $A$ satisfies the block-RIP of order  $K$ if there exist parameters $\delta_{K|\mathcal{I}}\in [0,\ 1)$ such that
\begin{eqnarray*}
  (1-\delta_{K|\mathcal{I}})\|x\|_2^2\leq\|Ax\|_2^2\leq(1+\delta_{K|\mathcal{I}})\|x\|_2^2
\end{eqnarray*}
for all block $K-$sparse signals $x$ over $\mathcal{I}$, where the smallest constant $\delta_{K|\mathcal{I}}$ is called as the block restricted isometry constant (block-RIC) of $A$.  By abuse of notation, we use $\delta_K$
for the block-RIC $\delta_{K|\mathcal{I}}$ when it is clear from the context.

This paper focuses on the BOMMP algorithm firstly proposed in \cite{g25} and described in Table $1$, which is a natural extension of the  BOMP algorithm. The BOMP algorithm only selects
one correct block index at each iteration. However,  the BOMMP algorithm identifies
$N(N\geq 1)$ block indices which contain at least one correct block index from the block support of the block sparse signal $x$ per iteration.
In \cite{g25},  the block-RIC
\begin{eqnarray}\label{geh1}
  \delta_{K+(N-2)k+N}<\frac{1}{1+\sqrt{\frac{K-k+1}{N}}}, \ \ \ \  (1\leq k \leq K)
\end{eqnarray}
is proved to be sufficient for the BOMMP to recover block $K-$sparse signals and  simulations  demonstrate the recovery performance of the BOMMP overtaking those of
the BOMP and BMP.

In this paper,  we provide a sharp sufficient condition of the reconstruction  of block  $K$-sparse signals through  the BOMMP.
In the noiseless case, we prove  that the condition with the block-RIC satisfying
\begin{eqnarray}\label{ghm1}
\delta_{NK+1}<\frac{1}{\sqrt{\frac{K}{N}+1}}
\end{eqnarray}
is sufficient to perfectly  recover  any block $K$-sparse signals via the BOMMP.  Moreover, we also prove that the sufficient condition \eqref{ghm1} is optimal, i.e., the for any given $K\in \N^+$, we construct a matrix $A$ satisfying
\begin{eqnarray*}
\delta_{NK+1}=\frac{1}{\sqrt{\frac{K}{N}+1}}
\end{eqnarray*}
such that the BOMMP may fail to recover  some block $K$-sparse signals $x$. Lastly,
 we also show $\delta_{NK+1}<\frac{1}{\sqrt{\frac{K}{N}+1}}$ together with a minimum $\ell_2$ norm of  nonzero blocks of the $K$-sparse signal $x$ can ensure the recovery of the support of $x$ through the BOMMP in noise case. If $N=1$, then the above condition \eqref{ghm1} is a sharp sufficient condition for the
recovery of block sparse signals by the BOMP \cite{g18}.
When $d_i=1\ (i=1,2,\cdots,l)$, the condition \eqref{ghm1} ensures that  the  gOMP or OMMP stably  recovers the sparse signal \cite{g26}, \cite{g27} and  is  also sharp \cite{g26}.
As $N=1$ and $d_i=1\ (i=1,2,\cdots,l)$, this condition \eqref{ghm1} turns to be a sharp sufficient condition for sparse  recovery  through OMP  \cite{16'}.

\begin{center} TABLE 1

The BOMMP algorithm
\end{center}
\hrule
\textbf{Input}\ \ \ \ \ \ \ \ \ measurements $y\in\R^m$, sensing matrix $A\in\R^{m\times n}$, block sparse level $K$,

\ \ \ \ \ \ \ \ \ \ \ \ \ \ number of indices for each selection $N$ $(N\leq K \ \textmd{and}\  N\leq\frac{m}{K})$.\\
\textbf{Initialize}\ \ \ \ iteration count $k=0$, residual vector $r^0=y$, estimated block support set

\ \ \ \ \ \ \ \ \ \ \ \ \ \ $\Lambda^0=\varnothing$.

\hrule
\textbf{While}\ \ \ \ \ $\|r^k\|_2>\epsilon$  and $k<\min\{K, \frac{m}{K}\}$ do \ $k=k+1$.

\ \ \ \ \ \ \ \ \ \ \ \ \ \ \ (Identification\ step)\ \ \ \ Select block indices set $T^k$ corresponding to $N$ largest

\ \ \ \ \ \ \ \ \ \ \ \ \ \ \  \ \ \ \ \ \ \ \ \ \ \ \ \ \ \ \ \ \ \ \ \ \ \ \ \ \ \ \ \ \  $\ell_2$ norm  of $\|A^{'}[i]r^{k-1}\|_2(i=1,2,\cdots,l)$.

\ \ \ \ \ \ \ \ \ \ \ \ \ \ \ (Augmentation step)\ \ \ $\Lambda^k=\Lambda^{k-1}\cup T^k$.

\ \ \ \ \ \ \ \ \ \ \ \ \ \  \ (Estimation step)\ \ \ \ \ \ \ \ $\hat{x}_{\Lambda^k}=\arg\min\limits_{u}\|y-A_{\Lambda^k}u\|_2$.

\ \ \ \ \ \ \ \ \ \ \ \ \ \  \ (Residual Update step)\ \ \ $r^k=y-A_{\Lambda^k}\hat{x}_{\Lambda^k}$.

\textbf{End}

\textbf{Output}\ \ \ \  the estimated signal
$\hat{x}=\arg\min\limits_{u:block-supp(u)=\Lambda^k}\|y-Au\|_2$. \hrule
\ \ \\

We begin, in Section \ref{2}, by
 giving some notations and  some basic lemmas that will be used. The main
results and their proofs are given in Section \ref{3}.

\section{Notations and lemmas}\label{preliminaries}\label{2}
 Throughout this paper, let $\Gamma\subseteq \{1,2,\ldots,l\}$ be a block index set and $\Gamma^c$ be the complementary set of $\Gamma$. Define a mixed $\ell_2/\ell_p$ norm with $p=1,2,\infty$
 as $\|x\|_{2,p}=\|w\|_p$, where $w\in\R^l$ with $w_i=\|x[i]\|_2$ for $i=1,2,\cdots,l$. Note that $\|x\|_{2,2}=\|x\|_2$. Let $\mathcal{I}_\Gamma=\{d_i: i\in \Gamma\}$ and  the block vector $x_\Gamma \in\R^{\sum_{i\in \Gamma}d_i}$ over $\mathcal{I}_\Gamma$  be  a concatenation of  $|\Gamma|$ blocks of length $d_i$($i\in \Gamma$). And let the block vector $\tilde{x}_\Gamma$ over $\mathcal{I}$ be  a concatenation of $l$ blocks of length $d_i$($i\in \mathcal{I}$)  satisfying
 \begin{eqnarray*}
\tilde{x}_\Gamma[i]= \left\{
   \begin{array}{ll}
     x_\Gamma[i], & \hbox{$i\in \Gamma$;} \\
     0\in\R^{d_i}, & \hbox{$i\in \{1,2,\ldots,l\}-\Gamma$,}
   \end{array}
 \right.
 \end{eqnarray*}
where $i=1,2,\cdots,l.$ Similarly,
Let $A_\Gamma $ over $\mathcal{I}_\Gamma$  be the submatrix of $A$, which is  a concatenation of  $|\Gamma|$ column blocks of length $d_i$($i\in \Gamma$).
Let $e_i\in\R^n$ be the $i$-th coordinate unit vector and $I_d$ be  the $d$-dimensional identity matrix, where $d$ is a positive integer.

Let $\alpha_N^{k+1}$ be the $N$-th largest  $\ell_2$ norm of $\|A^{'}[i]r^k\|_2$ with $i\in(T\cup\Lambda^k)^c$ and $\beta_1^{k+1}$ be the largest $\ell_2$ norm of $\|A^{'}[i]r^k\|_2$ with $i\in T-\Lambda^k$  in the $(k+1)$-th iteration of the BOMMP algorithm. Let $W_{k+1}\subseteq(T\cup\Lambda^k)^c$  be a set of $N$ block indices which correspond to
$N$ largest $\ell_2$ norm of $\|A^{'}[i]r^k\|_2$ with $i\in(T\cup\Lambda^k)^c$.

$A_{\Lambda^k}^\dagger$ represents the pseudo-inverse of $A_{\Lambda^k}$ when $A_{\Lambda^k}$ is full column rank ($\sum_{i\in \Lambda^k}d_i\leq m$), i.e., $A_{\Lambda^k}^\dagger=(A_{\Lambda^k}^{'}A_{\Lambda^k})^{-1}A_{\Lambda^k}^{'}$. Moreover, $P_{\Lambda^k}=A_{\Lambda^k}A_{\Lambda^k}^\dagger$ and $P^\bot_{\Lambda^k}=I-P_{\Lambda^k}$ denote two orthogonal projection operators which  project a given vector orthogonally onto the spanned space by all column blocks of $A_{\Lambda^k}$ and onto its orthogonal complement respectively.

First, we recall some useful lemmas in \cite{g18}.
 \begin{lem}\label{lem1}For any $K_1\leq K_2$,
 if the sensing matrix $A$ satisfies the block-RIP of order $K_2$, then $\delta_{K_1}\leq\delta_{K_2}$.
 \end{lem}
 \begin{lem}\label{lemm1}
  Let the  sensing matrix $A$ satisfy the block-RIP of order $K$ and $\Gamma$ be a block index set with $|\Gamma|\leq K$. Then
  there is
  \begin{eqnarray*}
  \|A^{'}_\Gamma x\|_2^2\leq(1+\delta_K)\|x\|_2^2
  \end{eqnarray*}
  for any $x\in\R^m$.
 \end{lem}

 Next, we will prove the  following lemma that plays an important role during  our analysis. It is rooted in \cite{16'} and \cite{g26}.
 \begin{lem}\label{lem2}
 For any nonempty index subset $W$  and any constants  $S,\ C>0$, let $t=\pm\frac{\sqrt{S+1}-1}{\sqrt{S}}$ and
 \begin{eqnarray}
 t_i=
         -\frac{C}{2}(1-t^2).
\end{eqnarray}
 Then for any vector $h_i\in\R^n$,  we have $t^2<1$ and
 \begin{eqnarray}\label{e1}
 \|A(x+\sum_{i\in W}t_ih_i)\|_2^2-\|A(t^2x-\sum_{i\in W}t_ih_i)\|_2^2=
 (1-t^4)\left(\langle Ax,Ax\rangle-C\sum_{i\in W}\langle Ax, Ah_i\rangle\right).\nonumber
 \end{eqnarray}
 \end{lem}
 \begin{proof}

 For $t=\pm\frac{\sqrt{S+1}-1}{\sqrt{S}}$, it follows that
 \begin{eqnarray}\label{e2}
 t^2=\frac{(\sqrt{S+1}-1)^2}{S}=\frac{\sqrt{S+1}-1}{\sqrt{S+1}+1}<1.\nonumber
 \end{eqnarray}
   By the following chain of equalities and the definition of $t_i\ (i\in W)$, we have that
 \begin{eqnarray}
 &&\|A(x+\sum_{i\in W}t_ih_i)\|_2^2-\|A(t^2x-\sum_{i\in W}t_ih_i)\|_2^2\nonumber\\
 &&=\langle Ax,Ax\rangle+2\sum_{i\in W}t_i\langle  Ax,Ah_i\rangle+2\sum_{i,j\in W,i\neq j}t_it_j\langle  Ah_i,Ah_j\rangle+\sum_{i\in W}t_i^2\langle  Ah_i,Ah_i\rangle\nonumber\\
 &&\ \ -\left(t^4\langle Ax,Ax\rangle-2t^2\sum_{i\in W}t_i\langle  Ax,Ah_i\rangle+2\sum_{i,j\in W,i\neq j}t_it_j\langle Ah_i,Ah_j\rangle+\sum_{i\in W}t_i^2\langle  Ah_i,Ah_i\rangle\right)\nonumber\\
 &&=(1-t^4)\langle Ax,Ax\rangle+2(1+t^2)\sum_{i\in W}t_i\langle  Ax,Ah_i\rangle\nonumber\\
 &&=(1-t^4)\left(\langle Ax,Ax\rangle-\frac{2}{1-t^2}(1-t^2)\frac{C}{2}\sum_{i\in W}\langle  Ax,Ah_i\rangle\right)\nonumber\\
 &&=(1-t^4)\left(\langle Ax,Ax\rangle-C\sum_{i\in W}\langle  Ax,Ah_i\rangle\right).\nonumber
 \end{eqnarray}
This completed the proof of  Lemma \ref{lem2}.
 \end{proof}

\section{ Main results }\label{3}

\subsection{Noiseless case}
\ \

It is clear that if $\beta^{k}_1>\alpha_N^k\ (1\leq k\leq K)$, then  at least one block index of $N$ block indices selected  is correct in every iteration, i.e., the BOMMP makes a success  in this iteration.
The following theorems provide a sufficient condition to guarantee the BOMMP algorithm success.
\begin{thm}\label{the1}
Suppose  $x$ is a block $K$-sparse signal and the sensing matrix $A$ satisfies the block-RIP of $K+N$ order with the block-RIC
\begin{eqnarray}\label{e34}
\delta_{K+N}<\frac{1}{\sqrt{\frac{K}{N}+1}}.
\end{eqnarray}
Then the BOMMP algorithm makes a success in the first iteration.
\end{thm}
\textbf{Remark}\ \textbf{1.} As $N=2$, the bound \eqref{e34} is
\begin{eqnarray*}
\delta_{K+2}<\frac{1}{\sqrt{\frac{K}{2}+1}}
\end{eqnarray*}
for the first iteration of the BOMMP.
In this case, \eqref{geh1} takes the form
\begin{eqnarray*}
  \delta_{K+2}<\frac{1}{1+\sqrt{\frac{K}{2}}}<\frac{1}{\sqrt{\frac{K}{2}+1}},
\end{eqnarray*}
that is,  the sufficient condition \eqref{e34} is weaker than that in \cite{g25} for the first iteration of the BOMMP.

\begin{proof} It is clear that we only need to consider the block $K$-sparse signal $x\neq 0$ in the proof.
Recall the definitions of $W_1$, $\alpha_N^1$ and $\beta_1^1$.  $W_1$ is a set of $N$ block indices which correspond to
$N$ largest $\ell_2$ norm of $\|A^{'}[i]r^k\|_2$ with $i\in T^c$.  $\alpha_N^{1}$ is the $N$-th largest  $\ell_2$ norm of $\|A^{'}[i]r^k\|_2$ with $i\in T^c$. $\beta_1^{1}$ be the largest $\ell_2$ norm of $\|A^{'}[i]r^k\|_2$ with $i\in T$.

Firstly, we consider $\alpha_{N}^{1}>0$, then $\|A^{'}[i]Ax\|_2>0$ for $\forall i\in W_1$.
Hence, we have that
\begin{eqnarray}\label{e4}
\alpha_N^1&=&\min\{\|A^{'}[i]Ax\|_2:i\in W_1\}\nonumber\\
&=&\min\{\langle A^{'}[i]Ax, \frac{A^{'}[i]Ax}{\|A^{'}[i]Ax\|_2}\rangle:i\in W_1\}\nonumber\\
&=&\min\{\langle Ax, A[i]a_{\{i\}}\rangle:i\in W_1\}\nonumber\\
&=&\min\{\langle Ax, A\widetilde{a}_{\{i\}}\rangle:i\in W_1\}\nonumber\\
&\leq&\frac{\sum_{i\in W_1} \langle Ax, A\widetilde{a}_{\{i\}}\rangle}{N},
\end{eqnarray}
where  $a_{\{i\}}=\frac{A^{'}[i]Ax}{\|A^{'}[i]Ax\|_2}$ with $\|a_{\{i\}}\|_2=1$.
It follows from the definition of $\beta_1^1$ and $|T|\leq K$ that
\begin{eqnarray}\label{e5}
\langle Ax, Ax \rangle&=&\langle \sum_{i\in T}A[i]x[i], Ax \rangle \nonumber\\
&=&\sum_{i\in T}\langle x[i], A^{'}[i]Ax \rangle\nonumber\\
&\leq &\sum_{i\in T}\|x[i]\|_2\|A^{'}[i]Ax\|_2\nonumber\\
&\leq&\beta_1^1\|x\|_{2,1}\nonumber\\
&\leq&\beta_1^1\sqrt{K}\|x\|_{2,2}\nonumber\\
&=&\beta_1^1\sqrt{K}\|x\|_{2}.
\end{eqnarray}

Let $t=-\frac{\sqrt{\frac{K}{N}+1}-1}{\sqrt{\frac{K}{N}}}$
and
\begin{eqnarray*}
 t_i=-\frac{\sqrt{K}}{2N}(1-t^2)\|x\|_{2}
\end{eqnarray*}
where $i\in W_1\subseteq T^c$ with $|W_1|=N$. Then we have that
\begin{eqnarray*}
t^2=\frac{\sqrt{\frac{K}{N}+1}-1}{\sqrt{\frac{K}{N}+1}+1}<1
\end{eqnarray*}
and
 \begin{eqnarray}\label{e33}
\sum_{i\in W_1}t_i^2&=&\left(\frac{\sqrt{K}}{2N}(1-t^2)\|x\|_{2}\right)^2N\nonumber\\
&=&\frac{K}{4N}(1-t^2)^2\|x\|_{2}^2\nonumber\\
&=&\frac{K}{4N}\left(1-\frac{\sqrt{\frac{K}{N}+1}-1}{\sqrt{\frac{K}{N}+1}+1}\right)^2\|x\|_{2}^2\nonumber\\
&=&\frac{K}{N}\frac{1}{\left(\sqrt{\frac{K}{N}+1}+1\right)^2}\|x\|_{2}^2\nonumber\\
&=&\frac{\sqrt{\frac{K}{N}+1}-1}{\sqrt{\frac{K}{N}+1}+1}\|x\|_{2}^2=t^2\|x\|_{2}^2.
 \end{eqnarray}
 From \eqref{e4}, \eqref{e5},  Lemma \ref{lem2} and $t^2<1$, it is clear that
\begin{eqnarray}\label{e6}
&&(1-t^4)\sqrt{K}\|x\|_{2}(\beta_1^1-\alpha_N^1)\nonumber\\
&&\geq(1-t^4)\left(\langle Ax, Ax \rangle-\sum_{i\in W_1}\sqrt{K}\|x\|_{2}\frac{\langle Ax,A\widetilde{a}_{\{i\}}\rangle}{N}\right)\nonumber\\
&&=\|A(x+\sum_{i\in W_1}t_i\widetilde{a}_{\{i\}})\|_2^2-\|A(t^2x-\sum_{i\in W_1}t_i\widetilde{a}_{\{i\}})\|_2^2.
\end{eqnarray}
Because the sensing matrix $A$ satisfies  the block-RIP of order $K+N$
with $\delta_{K+N}<\frac{1}{\sqrt{\frac{K}{N}+1}}$, $x\neq 0$ with the block-supp$(x)\subseteq T$ and $\|a_{\{i\}}\|_2=1$ with $i\in W_1\subseteq T^c$, it follows from \eqref{e33} that
\begin{eqnarray*}\label{e7}
&&\|A(x+\sum_{i\in W_1}t_i\widetilde{a}_{\{i\}})\|_2^2-\|A(t^2x-\sum_{i\in W_1}t_i\widetilde{a}_{\{i\}})\|_2^2\nonumber\\
&&\geq(1-\delta_{K+N})\left(\|x+\sum_{i\in W_1}t_i\widetilde{a}_{\{i\}}\|_2^2\right)-(1+\delta_{K+N})\left(\|t^2x-\sum_{i\in W_1}t_i\widetilde{a}_{\{i\}}\|_2^2\right)\nonumber\\
&&=(1-\delta_{K+N})\left(\|x\|_2^2+\sum_{i\in W_1}t_i^2\right)-(1+\delta_{K+N})\left(t^4\|x\|_2^2+\sum_{i\in W_1}t_i^2\right)\nonumber\\
&&=(1-\delta_{K+N})(1+t^2)\|x\|_{2}^2-(1+\delta_{K+N})(t^4+t^2)\|x\|_{2}^2\nonumber\\
&&=(1-t^4)\|x\|_{2}^2-\delta_{K+N}(1+t^2)^2\|x\|_{2}^2\nonumber\\
&&=(1+t^2)^2\|x\|_{2}^2\left(\frac{1-t^2}{1+t^2}-\delta_{K+N}\right)\nonumber\\
&&=(1+t^2)^2\|x\|_{2}^2\left(\frac{1-\frac{\sqrt{\frac{K}{N}+1}-1}{\sqrt{\frac{K}{N}+1}+1}}{1+\frac{\sqrt{\frac{K}{N}+1}-1}{\sqrt{\frac{K}{N}+1}+1}}-\delta_{K+N}\right)\nonumber\\
&&=(1+t^2)^2\|x\|_{2}^2\left(\frac{1}{\sqrt{\frac{K}{N}+1}}-\delta_{K+N}\right)\nonumber\\
&&>0.
\end{eqnarray*}
It follows from the above two  inequalities that $\beta_1^1>\alpha_N^1$,  which represents that the BOMMP algorithm selects  at least one block index from the block support $T$ under $\alpha_N^1>0$. As the above discussion, we have that $\beta_1^1>0$. When
$\alpha_N^1=0$, it is clear that $\beta_1^1>\alpha_N^1$.

As mentioned, if $\delta_{K+N}<\frac{1}{\sqrt{\frac{K}{N}+1}}$, then the BOMMP algorithm makes a success in the first iteration.
\end{proof}
\begin{thm}\label{the2}
Suppose the BOMMP algorithm  has performed $k$ iterations successfully, where $1\leq k<K$. Then the BOMMP algorithm will be successful for the $(k+1)$-th iteration if the sensing matrix $A$ satisfies the block-RIP of order $NK+1$ with the block-RIC $\delta_{NK+1}$ fulfilling
\begin{eqnarray*}
\delta_{NK+1}<\frac{1}{\sqrt{\frac{K}{N}+1}}.
\end{eqnarray*}
\end{thm}
\begin{proof}
For the BOMMP algorithm, $r^k=P^\perp_{\Lambda^k}y$ is orthogonal to each block of $A_{\Lambda^k}$ then
\begin{eqnarray*}
r^k&=&P^\perp_{\Lambda^k}y\\
&=&P^\perp_{\Lambda^k}A_Tx_T\\
&=&P^\perp_{\Lambda^k}(A_{T-{\Lambda^k}}x_{T-{\Lambda^k}}+A_{T\cap \Lambda^k}x_{T\cap\Lambda^k})\\
&=& P^\perp_{\Lambda^k}A_{T-{\Lambda^k}}x_{T-{\Lambda^k}}\\
&=&A_{T-{\Lambda^k}}x_{T-{\Lambda^k}}-P_{\Lambda^k}A_{T-{\Lambda^k}}x_{T-{\Lambda^k}}\\
&=&A_{T-{\Lambda^k}}x_{T-{\Lambda^k}}-A_{\Lambda^k}z_{\Lambda^k}\\
&=&A_{T\cup \Lambda^k}\omega_{T\cup \Lambda^k},
\end{eqnarray*}
where we used the fact that $P_{\Lambda^k}A_{T-{\Lambda^k}}x_{T-{\Lambda^k}}\in span(A_{\Lambda^k})$, so $P_{\Lambda^k}A_{T-{\Lambda^k}}x_{T-{\Lambda^k}}$ can be written as $A_{\Lambda^k}z_{\Lambda^k}$ for some $z_{\Lambda^k}\in \R^{\sum_{i\in \Lambda_k}d_i}$ and $\omega_{T\cup\Lambda^k}$ is given by
\begin{eqnarray*}
\omega_{T\cup \Lambda^k}=\left(
                              \begin{array}{c}
                                x_{T-\Lambda^k} \\
                                -z_{\Lambda^k}\\
                              \end{array}
                            \right).
\end{eqnarray*}

For the $(k+1)$-th iteration, if $T-\Lambda^k=\varnothing$,  then  $T \subseteq \Lambda^k$. Hence, the original block $K$-sparse signal
$x$ has already been recovered exactly. As  $T-\Lambda^k\neq\varnothing$, then $\omega_{T\cup \Lambda^k}\neq 0$. In the remainder of the proof,
we  consider firstly  $\alpha_N^{k+1}>0$,  then $\|A^{'}[i]r^k\|_2>0$ for $\forall i\in W_{k+1}$.
 We take $a_{\{i\}}=\frac{A^{'}[i]r^k}{\|A^{'}[i]r^k\|_2}=\frac{A^{'}[i]A_{T\cup \Lambda^k}\omega_{T\cup \Lambda^k}}{\|A^{'}[i]A_{T\cup \Lambda^k}\omega_{T\cup \Lambda^k}\|_2}$, then $\|\alpha_{\{i\}}\|_2=1$. In view of the definition of $\alpha_N^{k+1}$, we have that
\begin{eqnarray}\label{e9}
\alpha_N^{k+1}&=&\min\{\|A^{'}[i]r^k\|_2:i\in W_{k+1}\}\nonumber\\
&=&\min\{\langle A^{'}[i]r^k, \frac{A^{'}[i]r^k}{\|A^{'}[i]r^k\|_2}\rangle:i\in W_{k+1}\}\nonumber\\
&=&\min\{\langle r^k, A[i]a_{\{i\}}\rangle:i\in W_{k+1}\}\nonumber\\
&=&\min\{\langle A_{T\cup \Lambda^k}\omega_{T\cup \Lambda^k}, A\widetilde{a}_{\{i\}}\rangle:i\in W_{k+1}\}\nonumber\\
&\leq&\frac{ \sum_{i\in W_{k+1}}\langle A\widetilde{\omega}_{T\cup \Lambda^k}, A\widetilde{a}_{\{i\}}\rangle}{N}.
\end{eqnarray}
Combining the definition of $\beta_1^{k+1}$ with $A'_{\Lambda^k}r^k=0$, we derive that
\begin{eqnarray}\label{e10'}
\beta_1^{k+1}&=&\|A^{'}_{{T-\Lambda^k}}r^k\|_{2,\infty}\nonumber\\
&=&\|[A_{T-\Lambda^k}\ A_{T\cap\Lambda^k}]^{'}A_{T\cup\Lambda^k}\omega_{T\cup\Lambda^k}\|_{2,\infty}\nonumber\\
&=&\|A^{'}_{T}A_{T\cup\Lambda^k}\omega_{T\cup\Lambda^k}\|_{2,\infty}\nonumber\\
&=&\|[A_{T}A_{\Lambda^k-T}]^{'}A_{T\cup\Lambda^k}\omega_{T\cup\Lambda^k}\|_{2,\infty}\nonumber\\
&=&\|A^{'}_{T\cup \Lambda^k}A_{T\cup\Lambda^k}\omega_{T\cup\Lambda^k}\|_{2,\infty}.
\end{eqnarray}
Notice the fact that
\begin{eqnarray}\label{e16}
&&\|A^{'}_{T}A_{T\cup\Lambda^k}\omega_{T\cup\Lambda^k}\|_{2,\infty}\nonumber\\
&&\geq\frac{1}{\sqrt{K}}\|A^{'}_{T}A_{T\cup\Lambda^k}\omega_{T\cup\Lambda^k}\|_{2,2}\nonumber\\
&&=\frac{1}{\sqrt{K}}\|A^{'}_{T\cup\Lambda^k}A_{T\cup\Lambda^k}\omega_{T\cup\Lambda^k}\|_{2,2}\nonumber\\
&&=\frac{1}{\sqrt{K}}\|A^{'}_{T\cup\Lambda^k}A_{T\cup\Lambda^k}\omega_{T\cup\Lambda^k}\|_{2}.
\end{eqnarray}
From  \eqref{e10'} and \eqref{e16},  it follows that
\begin{eqnarray}\label{e11}
\langle A\widetilde{\omega}_{T\cup\Lambda^k}, A\widetilde{\omega}_{T\cup\Lambda^k}\rangle&=&\langle A_{T\cup\Lambda^k}\omega_{T\cup\Lambda^k}, A_{T\cup\Lambda^k}\omega_{T\cup\Lambda^k}\rangle\nonumber\\
&=&\langle A^{'}_{T\cup\Lambda^k}A_{T\cup\Lambda^k}\omega_{T\cup\Lambda^k}, \omega_{T\cup\Lambda^k}\rangle \nonumber\\
&\leq& \|A^{'}_{T\cup\Lambda^k}A_{T\cup\Lambda^k}\omega_{T\cup\Lambda^k}\|_2\|\omega_{T\cup\Lambda^k}\|_2\nonumber\\
&\leq&\sqrt{K}\beta_1^{k+1}\|\omega_{T\cup \Lambda^k}\|_2\nonumber\\
&=&\sqrt{K}\beta_1^{k+1}\|\widetilde{\omega}_{T\cup \Lambda^k}\|_2.
\end{eqnarray}

Similarly to the proof of Theorem \ref{the1},
let $t=-\frac{\sqrt{\frac{K}{N}+1}-1}{\sqrt{\frac{K}{N}}}$ and
\begin{eqnarray*}
 t_i=-\frac{\sqrt{K}}{2N}(1-t^2)\|\widetilde{\omega}_{T\cup \Lambda^k}\|_2,\ \ \ \ i\in W_{k+1}\subseteq (\Lambda^k\cup T)^c.
\end{eqnarray*}
 By \eqref{e9}, \eqref{e11} and Lemma \ref{lem2}, we have that
\begin{eqnarray}\label{e12}
&&(1-t^4)\sqrt{K}\|\widetilde{\omega}_{T\cup \Lambda^k}\|_2(\beta_1^{k+1}-\alpha_N^{k+1})\nonumber\\
&&\geq(1-t^4)\left(\langle A\widetilde{\omega}_{T\cup\Lambda^k}, A\widetilde{\omega}_{T\cup\Lambda^k}\rangle-\sqrt{K}\|\widetilde{\omega}_{T\cup \Lambda^k}\|_2\frac{\sum_{i\in W_{k+1}}\langle A\widetilde{\omega}_{T\cup \Lambda^k},  A\widetilde{a}_{\{i\}}\rangle}{N}\right)\nonumber\\
&&=\|A(\widetilde{\omega}_{T\cup\Lambda^k}+\sum_{i\in W_{k+1}}t_i a_{\{i\}})\|_2^2
-\|A(t^2\widetilde{\omega}_{T\cup\Lambda^k}-\sum_{i\in W_{k+1}}t_ia_{\{i\}})\|_2^2.
\end{eqnarray}
Let $l=|T\cap\Lambda^k|$, then $k\leq l\leq K$ and $Nk+K-l+N\leq NK+1$.  Since $A$ satisfies the block-RIP of order $NK+1$
with the  block-RIC $\delta_{NK+1}$, $\widetilde{\omega}_{{T\cup \Lambda^k}}\neq 0$ with the block-supp$(\widetilde{\omega}_{{T\cup \Lambda^k}})\subseteq T\cup \Lambda^k$ and $\|a_{\{i\}}\|_2=1$ with $i\in W_{k+1}\subseteq (T\cup\Lambda^k)^c$, it follows from Lemma \ref{lem1} and $\sum_{i\in W_{k+1}}t_i^2=t^2\|\widetilde{\omega}_{{T\cup \Lambda^k}}\|_2^2$ that
\begin{eqnarray*}\label{e7}
&&\|A(\widetilde{\omega}_{T\cup\Lambda^k}+\sum_{i\in W_{k+1}}t_i\widetilde{a}_{\{i\}})\|_2^2-\|A(t^2\widetilde{\omega}_{T\cup\Lambda^k}-\sum_{i\in W_{k+1}} t_i\widetilde{a}_{\{i\}})\|_2^2\nonumber\\
&&\geq(1-\delta_{Nk+K-l+N})\left(\|\widetilde{\omega}_{T\cup\Lambda^k}+\sum_{i\in W_{k+1}}t_i\widetilde{a}_{\{i\}}\|_2^2\right)\nonumber\\
&&\ \ \ \ -(1+\delta_{Nk+K-l+N})\left(\|t^2\widetilde{\omega}_{T\cup\Lambda^k}-\sum_{i\in W_{k+1}}t_i\widetilde{a}_{\{i\}}\|_2^2\right)\nonumber\\
&&=(1-\delta_{Nk+K-l+N})\left(\|\widetilde{\omega}_{T\cup\Lambda^k}\|_2^2+\sum_{i\in W_{k+1}}t_i^2\right)\nonumber\\
&&\ \ \ \ -(1+\delta_{Nk+K-l+N})\left(t^4\|\widetilde{\omega}_{T\cup\Lambda^k}\|_2^2+\sum_{i\in W_{k+1}}t_i^2\right)\nonumber\\
&&=(1-\delta_{Nk+K-l+N})(1+t^2)\|\widetilde{\omega}_{T\cup\Lambda^k}\|_2^2-(1+\delta_{Nk+K-l+N})(t^4+t^2)\|\widetilde{\omega}_{T\cup\Lambda^k}\|_2^2\nonumber\\
&&=(1+t^2)^2\|\widetilde{\omega}_{T\cup\Lambda^k}\|_2^2\left(\frac{1-t^2}{1+t^2}-\delta_{Nk+K-l+N}\right)\nonumber\\
&&\geq(1+t^2)^2\|\widetilde{\omega}_{T\cup\Lambda^k}\|_2^2\left(\frac{1-t^2}{1+t^2}-\delta_{NK+1}\right)\nonumber.
\end{eqnarray*}
Combining the fact that
\begin{eqnarray*}\label{e8}
\frac{1-t^2}{1+t^2}=\frac{1}{\sqrt{\frac{K}{N}+1}}
\end{eqnarray*}
with   the condition $\delta_{NK+1}<\frac{1}{\sqrt{\frac{K}{N}+1}}$,  it follows from $t^2<1$ and $\widetilde{\omega}_{T\cup\Lambda^k}\neq 0$ that
\begin{eqnarray*}
 (1-t^4)\sqrt{K}(\beta_1^{k+1}-\alpha_N^{k+1})
&\geq&(1+t^2)^2\left(\frac{1-t^2}{1+t^2}-\delta_{NK+1}\right)\|\widetilde{\omega}_{T\cup\Lambda^k}\|_2\nonumber\\
&\geq&(1+t^2)^2\left(\frac{1}{\sqrt{\frac{K}{N}+1}}-\delta_{NK+1}\right)\|\widetilde{\omega}_{T\cup\Lambda^k}\|_2\nonumber\\
&>&0,
\end{eqnarray*}
i.e., $\beta_1^{k+1}>\alpha_N^{k+1}$,  which ensures that the set $\Lambda^{k+1}$ contains at least one correct block index in the $(k+1)$-th iteration of the BOMMP algorithm under $\alpha_N^{k+1}>0$.
For $\alpha_N^{k+1}=0$, it is obvious that $\beta_1^{k+1}>\alpha_N^{k+1}$ based on $\omega_{T\cup\Lambda^k}\neq 0$.
 We have completed the proof of the theorem.
\end{proof}
Now combining the conditions for success in the first iteration in Theorem \ref{the1} with that in non-initial iterations in Theorem \ref{the2}, we obtain overall sufficient condition to guarantee the perfect recovery of block $K$-sparse signals  via  the BOMMP algorithm  in the following theorem.
\begin{thm}\label{the3}
Suppose  $x$ is a block $K$-sparse signal and the sensing matrix $A$ satisfies the block-RIP of order $NK+1$ with the block-RIC $\delta_{NK+1}$ fulfilling
\begin{eqnarray*}
\delta_{NK+1}<\frac{1}{\sqrt{\frac{K}{N}+1}}.
\end{eqnarray*}
Then the BOMMP algorithm can recover the block sparse signal $x$ exactly from $y=Ax$.
\end{thm}
\begin{proof}
For $N\geq1,\  K\geq1$ and $N\leq \min\{K,\ \frac{m}{K}\}$, then $K+N\leq NK+1$. It follows from Lemma \ref{lem1} that
\begin{eqnarray*}
\delta_{K+N}\leq\delta_{NK+1}<\frac{1}{\sqrt{\frac{K}{N}+1}}.
\end{eqnarray*}
Therefore, under the sufficient condition
$\delta_{NK+1}<\frac{1}{\sqrt{\frac{K}{N}+1}}$,  the BOMMP algorithm can recover perfectly any block $K$-sparse signals  from $y=Ax$ based on  Theorems \ref{the1} and \ref{the2}.
\end{proof}

Next, we prove that the proposed bound $\delta_{NK+1}<\frac{1}{\sqrt{\frac{K}{N}+1}}$ is optimal.
\begin{thm}\label{the4}
For any given $K\in \N^+$, there are a block $K$-sparse signal $x$ and a matrix $A$ satisfying
\begin{eqnarray*}
\delta_{NK+1}=\frac{1}{\sqrt{\frac{K}{N}+1}}
\end{eqnarray*}
such that the BOMMP may fail.
\end{thm}
In order to prove Theorem \ref{the4}, for a   positive integer $d$,  we firstly investigate the following  matrix $A(d)\in \R^{(NK+1)d\times (NK+1)d}$.
\begin{eqnarray*}
A(d)=
\begin{pmatrix}
  \ & \  & \  & 0 & \cdots & 0 & \frac{1}{b}I_d & \cdots & \frac{1}{b}I_d \\
  \  & \sqrt{\frac{K}{K+N}}I_{dK} & \  & \vdots & \vdots & \vdots & \vdots & \vdots & \vdots\\
  \  & \ & \  & 0 & \cdots & 0 & \frac{1}{b}I_d& \cdots & \frac{1}{b}I_d \\
  0 & \cdots & 0 & \  & \  & \  & 0 & \cdots& 0 \\
  \vdots & \vdots & \vdots & \  & I_{d(NK+1-N-K)} & \  & \vdots & \vdots & 0\\
  0 & \cdots & 0 & \  & \  & \  & 0 & \cdots & 0 \\
  0 & \cdots & 0& 0 & \cdots & 0  & \  & \  & \  \\
  \vdots & \vdots & \vdots & \vdots & \vdots & \vdots & \  & I_{dN} & \  \\
  0 & \cdots & 0& 0 & \cdots & 0 & \  & \  & \  \\
\end{pmatrix},
\end{eqnarray*}
where $b=\sqrt{K(K+N)}$. Then we have that
\begin{eqnarray*}
A'(d)A(d)=\begin{pmatrix}
  \ & \  & \  & 0 & \cdots & 0 & sI_d & \cdots & s I_d\\
  \  & \frac{K}{K+N}I_{dK} & \  & \vdots & \vdots & \vdots & \vdots & \vdots & \vdots\\
  \  & \ & \  & 0 & \cdots & 0 & sI_d & \cdots & sI_d \\
  0 & \cdots & 0 & \  & \  & \  &  0 & \cdots& 0 \\
  \vdots & \vdots & \vdots & \  & I_{d(NK+1-N-K)} & \  & \vdots &\vdots &  \vdots\\
  0 & \cdots & 0 & \  & \  & \  & 0 & \cdots & 0 \\
  s I_d& \cdots & sI_d & 0 & \cdots & 0 & (1+s)I_d & \cdots &sI_d\\
  \vdots & \vdots & \vdots & \vdots & \vdots & \vdots & \vdots& \ddots & \vdots \\
  sI_d & \cdots & sI_d& 0 & \cdots & 0 & sI_d& \cdots&(1+s)I_d \\
\end{pmatrix},
\end{eqnarray*}
where $s=\frac{1}{K+N}$. By elementary transformation of determinant, we have that
\begin{eqnarray}\label{gh1}
&&\begin{vmatrix}
 A'(d)A(d)-\lambda I_{(NK+1)d} \\
\end{vmatrix}\nonumber\\
&&=\begin{pmatrix}
  \ & \  & \  & 0 & \cdots & 0 & sI_d & \cdots & Ns I_d\\
  \  & \ & \  & 0 & \cdots & 0 & 0 & \cdots & 0 \\
  \  & s_1I_{dK} & \  & \vdots & \vdots & \vdots & \vdots & \vdots & \vdots\\
  \  & \ & \  & 0 & \cdots & 0 & 0 & \cdots & 0 \\
  0 & \cdots & 0 & \  & \  & \  & 0 & \cdots& 0 \\
  \vdots & \vdots & \vdots & \  & s_2I_{d(NK+1-N-K)} & \  & \vdots & \vdots & 0\\
  0 & \cdots & 0 & \  & \  & \  & 0 & \cdots & 0 \\
  0& \cdots & 0 & 0 & \cdots & 0 & \  & \cdots &0\\
  \vdots & \vdots & \vdots & \vdots & \vdots & \vdots &\  & s_2I_{(N-1)d} & \vdots \\
  0& \cdots & 0 & 0 & \cdots & 0 & \   &  &0\\
  KsI_d & \cdots & sI_d& 0 & \cdots & 0 & sI_d& \cdots&s_3I_d \\
\end{pmatrix}
\end{eqnarray}
where $s_1=\frac{K}{K+N}-\lambda,\ s_2=1-\lambda$ and $s_3=1+\frac{N}{K+N}-\lambda$. Next, we claim that
\begin{eqnarray}\label{gh2}
  &&\begin{vmatrix}
 A'(d)A(d)-\lambda I_{(NK+1)d} \\
\end{vmatrix}\nonumber\\
&&=(1-\lambda)^{d(NK-K)}(\frac{K}{K+N}-\lambda)^{d(K-1)}(\lambda^2-2\lambda+\frac{K}{K+N})^{d}.
\end{eqnarray}
By inductive mwthod, we prove the above claim \eqref{gh2}.
As $d=1$, by direct calculation,  it follows from \eqref{gh1} that
\begin{eqnarray*}
&&\begin{vmatrix}
  A'(1)A(1)-\lambda I_{NK+1} \\
\end{vmatrix}=(1-\lambda)^{NK-K}(\frac{K}{K+N}-\lambda)^{K-1}(\lambda^2-2\lambda+\frac{K}{K+N}).
\end{eqnarray*}
For $d-1(d\geq 2)$,
suppose
\begin{eqnarray*}
&&\begin{vmatrix}
  A'(d-1)A(d-1)-\lambda I_{(d-1)(NK+1)} \\
\end{vmatrix}\\
&&=(1-\lambda)^{(d-1)(NK-K)}(\frac{K}{K+N}-\lambda)^{(d-1)(K-1)}(\lambda^2-2\lambda+\frac{K}{K+N})^{(d-1)}.
\end{eqnarray*}
For $d\geq2$, we expand the determinant \eqref{gh1} by the first column, then expand the remaining determinant by the first row of $s_1I_d$, $s_2I_d$ and $s_3I_d$.
 Hence, we have that
 \begin{eqnarray*}
&&\begin{vmatrix}
  A'(d)A(d)-\lambda I \\
\end{vmatrix}_{d(NK+1)}\\
&&=(-1)^{(1+1)}(\frac{K}{K+N}-\lambda)\left((-1)^{(d-1)+1+(d-1)+1}(\frac{K}{K+N}-\lambda)\cdots\right.
\\
&&\ \ \ \ \ (-1)^{(K-1)(d-1)+1+(K-1)(d-1)+1}(\frac{K}{K+N}-\lambda)
(-1)^{K(d-1)+1+K(d-1)+1}(1-\lambda)\cdots\\
&&\ \ \ \ \ (-1)^{(NK-1)(d-1)+1+(NK-1)(d-1)+1}(1-\lambda)(-1)^{NK(d-1)+1+NK(d-1)+1}(1+\frac{N}{K+N}-\lambda)\\
&&\ \ \ \ \ \left.\begin{vmatrix}
  A'(d-1)A(d-1)-\lambda I_{(d-1)(NK+1)} \\
\end{vmatrix}\right)\\
&&\ \ \ +(-1)^{dNK+1+1}\frac{K}{K+N}\left((-1)^{d+1+(d-1)+1}(\frac{K}{K+N}-\lambda)\cdots\right.\\
&&\ \ \ \ \ (-1)^{d+(K-2)(d-1)+1+(K-1)(d-1)+1}(\frac{K}{K+N}-\lambda)
(-1)^{d+(K-1)(d-1)+1+K(d-1)+1}(1-\lambda)\cdots\\
&&\ \ \ \ \ (-1)^{d+(NK-2)(d-1)+1+(NK-1)(d-1)+1}(1-\lambda)(-1)^{NK(d-1)+1+1}\frac{N}{K+N}\\
&&\ \ \ \ \ \left.\begin{vmatrix}
  A'(a)A(a)-\lambda I_{(d-1)(NK+1)} \\
\end{vmatrix}\right)\\
&&=(\frac{K}{K+N}-\lambda)\left((\frac{K}{K+N}-\lambda)^{K-1}
(1-\lambda)^{NK+1-K-N+N-1}(1+\frac{N}{K+N}-\lambda)\right.\\
&&\ \ \ \ \ \left.\begin{vmatrix}
  A'(d-1)A(d-1)-\lambda I_{(d-1)(NK+1)} \\
\end{vmatrix}\right)\\
&&\ \ \ -\frac{K}{K+N}\left((\frac{K}{K+N}-\lambda)^{K-1}(1-\lambda)^{NK+1-K-N+N-1}\frac{N}{K+N}\right.\\
&&\ \ \ \ \ \left.\begin{vmatrix}
  A'(d-1)A(d-1)-\lambda I_{(d-1)(NK+1)} \\
\end{vmatrix}\right)\\
&&=(1-\lambda)^{d(NK-K)}(\frac{K}{K+N}-\lambda)^{d(K-1)}(\lambda^2-2\lambda+\frac{K}{K+N})^{d}.
\end{eqnarray*}
Therefore, we have completed the proof of the claim \eqref{gh2}.

Now, we present the proof of Theorem \ref{the4}.
\begin{proof}
 For convenience, we assume that the block $K$-sparse signal $x$ consists of $NK+1$ blocks each having identical length of $d$, i.e., $n=(NK+1)d$.
For any given positive integer $K$, let $A=A(d)$.
By \eqref{gh2}, it is clear that
$\frac{K}{K+N}$, $1$, $1-\frac{1}{\sqrt{\frac{K}{N}+1}}$ and $1+\frac{1}{\sqrt{\frac{K}{N}+1}}$ are eigenvalues of $A'A$ with multiplicity of $d(K-1)$, $d(NK-K)$, $d$ and $d$ respectively.  Moreover, $1-\frac{1}{\sqrt{\frac{K}{N}+1}}$ and $1+\frac{1}{\sqrt{\frac{K}{N}+1}}$ are the minimum and maximum eigenvalue  of $A'A$  respectively.

So for $\forall x \in \R^{(NK+1)d}$, we easily derive that
\begin{eqnarray*}
(1-\frac{1}{\sqrt{\frac{K}{N}+1}})\|x\|_2^2\leq x'A'Ax\leq (1+\frac{1}{\sqrt{\frac{K}{N}+1}})\|x\|_2^2,
\end{eqnarray*}
i.e.,
\begin{eqnarray*}
(1-\frac{1}{\sqrt{\frac{K}{N}+1}})\|x\|_2^2\leq \|Ax\|_2^2\leq (1+\frac{1}{\sqrt{\frac{K}{N}+1}})\|x\|_2^2,
\end{eqnarray*}
Therefore, we have that
\begin{eqnarray*}
 \delta_{NK+1}\leq \frac{1}{\sqrt{\frac{K}{N}+1}}.
 \end{eqnarray*}
Next, we claim that the matrix $A$ satisfies the block-RIP of order $NK + 1$ with the block-RIC
\begin{eqnarray*}
 \delta_{NK+1}=\frac{1}{\sqrt{\frac{K}{N}+1}}.
 \end{eqnarray*}
Let $h\in \R^{NK+1}$ be the eigenvector of $A'(1)A(1)$ corresponding to the eigenvalue $1+\frac{1}{\sqrt{\frac{K}{N}+1}}$ and $x\in \R^{(NK+1)d}$ with
$x[i]=h_ie_1$($e_1\in \R^{d}$ is the first coordinate unit vector) for $1\leq i\leq NK+1$. Then we obtain that
\begin{eqnarray*}
  x'A'Ax=h'A'(1)A(1)h=(1+\frac{1}{\sqrt{\frac{K}{N}+1}})\|h\|_2^2=(1+\frac{1}{\sqrt{\frac{K}{N}+1}})\|x\|_2^2.
\end{eqnarray*}
Therefore $A$ satisfies the block-RIC $\delta_{NK+1}=\frac{1}{\sqrt{\frac{K}{N}+1}}$.

 Consider the block  $K$-sparse signal $x=(e_1,e_1,\cdots,e_1,0\cdots,0)'\in\R^{(NK+1)d}$, i.e., $T=$block-supp$(x)=\{1,2,\cdots,K\}$.
For the first iteration, there are
\begin{eqnarray}\label{eq1}
 \| A'[i]r^0\|_2=\| A'[i]Ax\|_2=\left\{
                                                                \begin{array}{ll}
                                                                  \frac{K}{K+N}, & \hbox{$i\in T$;} \\
                                                                  0, & \hbox{$i\in \{K+1, \cdots, NK+1-N\}$;} \\
                                                                  \frac{K}{K+N}, & \hbox{$i\in \{NK+2-N, \cdots, NK+1\}$.}
                                                                \end{array}
                                                              \right.
\end{eqnarray}

Therefore, it follows from the definitions of $\beta_1^1$ and $\alpha_N^1$ and \eqref{eq1} that $\beta_1^1=\frac{K}{K+N}$ and $\alpha_N^1=\frac{K}{K+N}$, that is, $\beta_1^1=\alpha_N^1$.
 This implies the BOMMP may fail to identify at least one correct index in the first iteration. So the BOMMP algorithm may fail for the given matrix $A$ and the block $K$-sparse  signal $x$.
 \end{proof}

\subsection{Noise case}
\ \

In the subsection, we show that  a high order block-RIP condition can guarantee stable and robust recovery
of all block $K$-sparse signals in bounded $\ell_2$ noise setting via the BOMMP algorithm from $y=Ax+e$.
 A sufficient condition in terms of the block-RIC $\delta_{NK+1}$ and the minimum $\ell_2$ norm of  nonzero blocks of block $K$-sparse signals $x$ is described as follow.
\begin{thm}\label{the5}
Suppose $\|e\|_2\leq\varepsilon$ and  the sensing matrix $A$ satisfies a high order block-RIP with the block-RIC
\begin{eqnarray}\label{e13}
\delta_{NK+1}<\frac{1}{\sqrt{\frac{K}{N}+1}}.
\end{eqnarray}
Then the BOMMP algorithm  with the stopping rule $\|r^k\|_2\leq\varepsilon$ recovers exactly the correct support of  block $K$-sparse signals $x$
if  all the nonzero blocks $x[i]$ satisfy
\begin{eqnarray}\label{e14}
\|x[i]\|_2>\max\bigg\{\frac{\frac{\sqrt{2 K(1+\delta_{NK+1})}\varepsilon}{\sqrt{\frac{K}{N}+1}}}{\frac{1}{\sqrt{\frac{K}{N}+1}}-\delta_{NK+1}},\frac{2\varepsilon}
{\sqrt{1-\delta_{NK+1}}}\bigg\}.
\end{eqnarray}
\end{thm}
\begin{proof} Use mathematical induction method to prove the theorem. Suppose the BOMMP performed $k(1\leq k\leq K-1)$ iterations successfully.
Now considering the $(k+1)$-th iteration, we have that
\begin{eqnarray*}
r^k&=&P^\perp_{\Lambda^k}y\\
&=&P^\perp_{\Lambda^k}A_{T}x_{T}+P_{\Lambda^k}^\perp e\\
&=&A_{T\cup \Lambda^k}\omega_{T\cup \Lambda^k}+(I-P_{\Lambda^k})e
\end{eqnarray*}
for some $\omega_{T\cup \Lambda^k}$ as in the proof of Theorem \ref{the2}.
One consider the following two cases.
\begin{itemize}
\item Case $1$: $T-\Lambda^k=\varnothing$
\end{itemize}

This implies  $T\subseteq\Lambda^k$.
Then  the correct support $T$ of the original block $K$-sparse signal $x$ has already been  recovered.
\begin{itemize}
\item Case $2$: $T-\Lambda^k\neq\varnothing$, i.e., $|T-\Lambda^k|\geq1$
\end{itemize}
In this case, it is clear that $\omega_{T\cup \Lambda^k}\neq 0$.
 Without loss of generality,  we only  consider  $\alpha^{k+1}_N>0$,
then $\|A'[i]r^{k}\|_2>0$ for $\forall i \in W_{k+1}\subseteq (T\cup \Lambda^k)^c$.
In the following proof, we take $a_{\{i\}}=\frac{A^{'}[i]A_{T\cup \Lambda^k}\omega_{T\cup \Lambda^k}}{\|A^{'}[i]A_{T\cup \Lambda^k}\omega_{T\cup \Lambda^k}\|_2}\ (i\in W_{k+1})$, then with $\|a_{\{i\}}\|_2=1$.

 Using the definition of $\alpha^{k+1}_N$, we have that
\begin{eqnarray}\label{e15}
\alpha_N^{k+1}&=&\min\{\|A^{'}[i]r^k\|_2:i\in W_{k+1}\}\nonumber\\
&\leq&\min\{\|A^{'}[i]A_{T\cup \Lambda^k}\omega_{T\cup \Lambda^k}\|_2+\|A'[i](I-P_{\Lambda^k})e\|_2:i\in W_{k+1}\}\nonumber\\
&=&\min\{\langle A^{'}[i]A_{T\cup \Lambda^k}\omega_{T\cup \Lambda^k}, \frac{A^{'}[i]A_{T\cup \Lambda^k}\omega_{T\cup \Lambda^k}}{\|A^{'}[i]A_{T\cup \Lambda^k}\omega_{T\cup \Lambda^k}\|_2}\rangle\nonumber\\
&&\ \ +\|A'[i](I-P_{\Lambda^k})e\|_2:i\in W_{k+1}\}\nonumber\\
&=&\min\{\langle A_{T\cup \Lambda^k}\omega_{T\cup \Lambda^k}, A[i]a_{\{i\}}\rangle+\|A'[i](I-P_{\Lambda^k})e\|_2:i\in W_{k+1}\}\nonumber\\
&=&\min\{\langle A_{T\cup \Lambda^k}\omega_{T\cup \Lambda^k}, A\widetilde{a}_{\{i\}}\rangle+\|A'[i](I-P_{\Lambda^k})e\|_2:i\in W_{k+1}\}\nonumber\\
&\leq&\frac{\sum_{i\in W_{k+1}}\langle A\widetilde{\omega}_{T\cup \Lambda^k},  A\widetilde{a}_{\{i\}}\rangle+\sum_{i\in W_{K+1}}\|A'[i](I-P_{\Lambda^k})e\|_2}{N}.
\end{eqnarray}
By the definition of $\beta^{k+1}_1$ and the fact $A'_{\Lambda^k} r^k=0$,  it follows from \eqref{e16} and \eqref{e11} that
\begin{eqnarray}\label{e17}
\sqrt{K}\|\omega_{T\cup \Lambda^k}\|_2\beta_1^{k+1}&=&\sqrt{K}\|\omega_{T\cup \Lambda^k}\|_2\|A^{'}_{{T-\Lambda^k}}r^k\|_{2,\infty}\nonumber\\
&=&\sqrt{K}\|\omega_{T\cup \Lambda^k}\|_2\|[A_{T-\Lambda^k}\ A_{T\cap\Lambda^k}]^{'}r^k\|_{2,\infty}\nonumber\\
&=&\sqrt{K}\|\omega_{T\cup \Lambda^k}\|_2\|A^{'}_Tr^k\|_{2,\infty}\nonumber\\
&=&\sqrt{K}\|\omega_{T\cup \Lambda^k}\|_2\|[A_T\ A_{\Lambda^k-T}]^{'}r^k\|_{2,\infty}\nonumber\\
&=&\sqrt{K}\|\omega_{T\cup \Lambda^k}\|_2\|A^{'}_{T\cup \Lambda^k}r^k\|_{2,\infty}\nonumber\\
&\geq&\sqrt{K}\|\omega_{T\cup \Lambda^k}\|_2(\|A^{'}_{T\cup \Lambda^k}A_{T\cup \Lambda^k}\omega_{T\cup \Lambda^k}\|_{2,\infty}-\|A^{'}_{T\cup \Lambda^k}(I-P_{\Lambda^k})e\|_{2,\infty})\nonumber\\
&\geq&\|\omega_{T\cup \Lambda^k}\|_2\|A^{'}_{T\cup \Lambda^k}A_{T\cup \Lambda^k}\omega_{T\cup \Lambda^k}\|_{2}\nonumber\\
&&-\sqrt{K}\|\omega_{T\cup \Lambda^k}\|_2\|A^{'}_{T\cup \Lambda^k}(I-P_{\Lambda^k})e\|_{2,\infty}\nonumber\\
&\geq&\langle A\widetilde{\omega}_{T\cup \Lambda^k}, A\widetilde{\omega}_{T\cup \Lambda^k} \rangle-\sqrt{K}\|\omega_{T\cup \Lambda^k}\|_2\|A^{'}_{T\cup \Lambda^k}(I-P_{\Lambda^k})e\|_{2,\infty}.
\end{eqnarray}

Let $t=-\frac{\sqrt{\frac{K}{N}+1}-1}{\sqrt{\frac{K}{N}}}$
and
\begin{eqnarray*}
 t_i=-\frac{\sqrt{K}}{2N}(1-t^2)\|\omega_{T\cup \Lambda^k}\|_2,\ \ \ \ i\in W_{k+1}\subseteq ( T\cup\Lambda^k)^c.
\end{eqnarray*}
Then we have
\begin{eqnarray}\label{e18}
\sum_{i\in W_{k+1}}t_i^2=t^2\|\omega_{T\cup \Lambda^k}\|_2^2.
\end{eqnarray}
It follows from \eqref{e15}, \eqref{e17} and $t^2<1$ that
\begin{eqnarray}\label{eq3}
&&(1-t^4)\sqrt{K}\|\omega_{T\cup \Lambda^k}\|_2(\beta_1^{k+1}-\alpha_N^{k+1})\nonumber\\
&&\geq(1-t^4)\bigg(\langle A\widetilde{\omega}_{T\cup \Lambda^k}, A\widetilde{\omega}_{T\cup \Lambda^k} \rangle-\sqrt{K}\|\omega_{T\cup \Lambda^k}\|_2\|A^{'}_{T\cup \Lambda^k}(I-P_{\Lambda^k})e\|_{2,\infty}\bigg.\nonumber\\
&& \left.-\frac{\sqrt{K}\|\omega_{T\cup \Lambda^k}\|_2(\sum_{i\in W_{k+1}}\langle A\widetilde{\omega}_{T\cup \Lambda^k}, A\widetilde{a}_{\{i\}}\rangle+\sum_{i\in W_{K+1}}\|A'[i](I-P_{\Lambda^k})e\|_2)}{N}\right)\nonumber\\
&&=\|A(\widetilde{\omega}_{T\cup \Lambda^k}+\sum_{i\in W_{k+1}}t_i\widetilde{a}_{\{i\}})\|_2^2
-\|A(t^2\widetilde{\omega}_{T\cup \Lambda^k}-\sum_{i\in W_{k+1}}t_i\widetilde{a}_{\{i\}})\|_2^2-(1-t^4)\sqrt{K}\nonumber\\
&&\|\omega_{T\cup \Lambda^k}\|_2\left(\|A^{'}_{T\cup \Lambda^k}(I-P_{\Lambda^k})e\|_{2,\infty}
+\frac{\sum_{i\in W_{k+1}}\|A'[i](I-P_{\Lambda^k})e\|_2}{N}\right).
\end{eqnarray}
As in the proof of Theorem \ref{the2}, $l=|T\cap\Lambda^k|$ then $Nk+K-l+N\leq NK+1$.
Because  $A$ satisfies the block-RIP
with the block-RIC $\delta_{NK+1}$, $\widetilde{\omega}_{T\cup \Lambda^k}\neq 0$ with block-supp$(\widetilde{\omega}_{T\cup \Lambda^k})\subseteq T\cup \Lambda^k$, and $\|a_{\{i\}}\|_2=1$ for $i\in W_{k+1}\subseteq (T\cup\Lambda^k)^c$, it follows from \eqref{e18} and Lemma \ref{lem1} that
\begin{eqnarray}\label{e7}
&&\|A(\widetilde{\omega}_{T\cup \Lambda^k}+\sum_{i\in W_{k+1}}t_i\widetilde{a}_{\{i\}})\|_2^2
-\|A(t^2\widetilde{\omega}_{T\cup \Lambda^k}-\sum_{i\in W_{k+1}}t_i\widetilde{a}_{\{i\}})\|_2^2\nonumber\\
&&\geq(1-\delta_{Nk+K-l+N})\left(\|\widetilde{\omega}_{T\cup \Lambda^k}+\sum_{i\in W_{k+1}}t_i\widetilde{a}_{\{i\}}\|_2^2\right)\nonumber\\
&&\ \ -(1+\delta_{Nk+K-l+N})\left(\|t^2\widetilde{\omega}_{T\cup \Lambda^k}-\sum_{i\in W_{k+1}}t_i\widetilde{a}_{\{i\}}\|_2^2\right)\nonumber\\
&&=(1-\delta_{Nk+K-l+N})\left(\|\widetilde{\omega}_{T\cup \Lambda^k}\|_2^2+\sum_{i\in W_{k+1}}t_i^2\right)\nonumber\\
&&\ \ -(1+\delta_{Nk+K-l+N})\left(t^4\|\widetilde{\omega}_{T\cup \Lambda^k}\|_2^2+\sum_{i\in W_{k+1}}t_i^2\right)\nonumber\\
&&=(1-\delta_{Nk+K-l+N})\|\widetilde{\omega}_{T\cup \Lambda^k}\|_2^2(1+t^2)-(1+\delta_{Nk+K-l+N})\|\widetilde{\omega}_{T\cup \Lambda^k}\|_2^2(t^4+t^2)\nonumber\\
&&=(1-t^4)\|\widetilde{\omega}_{T\cup \Lambda^k}\|_2^2-\delta_{Nk+K-l+N}\|\widetilde{\omega}_{T\cup \Lambda^k}\|_2^2(1+t^2)^2\nonumber\\
&&=(1+t^2)^2\|\widetilde{\omega}_{T\cup \Lambda^k}\|_2^2\left(\frac{1-t^2}{1+t^2}-\delta_{Nk+K-l+N}\right)\nonumber\\
&&\geq(1+t^2)^2\|\widetilde{\omega}_{T\cup \Lambda^k}\|_2^2\left(\frac{1-t^2}{1+t^2}-\delta_{NK+1}\right).
\end{eqnarray}
Notice that there exist $i_k\in (T\cup\Lambda^k)$ and $ j_k\in (T\cup\Lambda^k)^c$ satisfying
\begin{eqnarray*}
  \|A^{'}_{T\cup \Lambda^k}(I-P_{\Lambda^k})e\|_{2,\infty}=\|A^{'}[i_k](I-P_{\Lambda^k})e\|_{2};\\
\|A^{'}_{(T\cup \Lambda^k)^c}(I-P_{\Lambda^k})e\|_{2,\infty}=\|A^{'}[j_k](I-P_{\Lambda^k})e\|_{2}.\\
\end{eqnarray*}
Hence,
\begin{eqnarray}\label{eq4}
&&\|A^{'}_{T\cup \Lambda^k}(I-P_{\Lambda^k})e\|_{2,\infty}+\|A^{'}_{(T\cup \Lambda^k)^c}(I-P_{\Lambda^k})e\|_{2,\infty}\nonumber \\
&&= \|A^{'}[i_k](I-P_{\Lambda^k})e\|_{2}+\|A^{'}[j_k](I-P_{\Lambda^k})e\|_{2}\nonumber\\
&&=\|A^{'}_{\{i_k \ j_k\}}(I-P_{\Lambda^k})e\|_{2,1}\nonumber\\
&&\leq\sqrt{2}\|A^{'}_{\{i_k \  j_k\}}(I-P_{\Lambda^k})e\|_{2}\nonumber\\
&&\leq\sqrt{2(1+\delta_{NK+1})}\|(I-P_{\Lambda^k})e\|_{2}\nonumber\\
&&\leq\sqrt{2(1+\delta_{NK+1})}\|e\|_{2}\nonumber\\
&&\leq\sqrt{2(1+\delta_{NK+1})}\varepsilon,
\end{eqnarray}
where we use Lemmas \ref{lem1} and \ref{lemm1} and the fact
\begin{eqnarray*}
  \|(I-P_{\Lambda^k})e\|_2\leq\|I-P_{\Lambda^k}\|_2\|e\|_2\leq\|e\|_2\leq\varepsilon.
\end{eqnarray*}
From \eqref{eq3}, \eqref{e7}, \eqref{eq4}, \eqref{e13} and \eqref{e14}, it follows that
\begin{eqnarray*}
&&(1-t^4)\sqrt{K}\|\widetilde{\omega}_{T\cup \Lambda^k}\|_2(\beta_1^{k+1}-\alpha_N^{k+1})\nonumber\\
&&\geq(1+t^2)^2\|\widetilde{\omega}_{T\cup \Lambda^k}\|_2^2\left(\frac{1-t^2}{1+t^2}-\delta_{NK+1}\right)-(1-t^4)\sqrt{K}\|\widetilde{\omega}_{T\cup \Lambda^k}\|_2\nonumber\\
&&\ \ \left(\|A^{'}_{T\cup \Lambda^k}(I-P_{\Lambda^k})e\|_{2,\infty}
+\frac{\sum_{i\in W_{k+1}}|\|A'[i] (I-P_{\Lambda^k})e\|_2}{N}\right)\nonumber\\
&&\geq(1+t^2)^2\|\widetilde{\omega}_{T\cup \Lambda^k}\|_2^2\left(\frac{1-t^2}{1+t^2}-\delta_{NK+1}\right)-(1-t^4)\sqrt{K}\|\widetilde{\omega}_{T\cup \Lambda^k}\|_2\sqrt{2(1+\delta_{NK+1})}\varepsilon\nonumber\\
&&=(1+t^2)^2\|\widetilde{\omega}_{T\cup \Lambda^k}\|_2\left(\left(\frac{1-t^2}{1+t^2}-\delta_{NK+1}\right)\|\widetilde{\omega}_{T\cup \Lambda^k}\|_2-\sqrt{2K(1+\delta_{NK+1})}\varepsilon\frac{1-t^2}{1+t^2}\right)\nonumber\\
&&\geq(1+t^2)^2\|\widetilde{\omega}_{T\cup \Lambda^k}\|_2\left(\left(\frac{1-t^2}{1+t^2}-\delta_{NK+1}\right)\|x_{T-\Lambda_k}\|_2-\sqrt{2K(1+\delta_{NK+1})}\varepsilon\frac{1-t^2}{1+t^2}\right)\nonumber\\
&&\geq(1+t^2)^2\|\widetilde{\omega}_{T\cup \Lambda^k}\|_2\left(\left(\frac{1}{\sqrt{\frac{K}{N}+1}}-\delta_{NK+1}\right)\sqrt{|T-\Lambda^k|}\min_{i\in T-\Lambda^k} \|x[i]\|_2\right.\nonumber\\
&&\ \ \ \left.-\frac{\sqrt{2K(1+\delta_{NK+1})}\varepsilon}{\sqrt{\frac{K}{N}+1}}\right)\nonumber\\
&&>0,
\end{eqnarray*}
i.e., $\beta_1^{k+1}>\alpha_N^{k+1}$ which guarantees at least one  index selected  from the correct support in the $(k+1)-$th iteration.

It remains to show that the BOMMP exactly stops under the stopping rule $\|r^k\|\leq\varepsilon$ when all the correct block indices are selected.
First, assume that $T-\Lambda^k=\varnothing$, then $T\subseteq \Lambda^k$ and $(I-P_{\Lambda^k})Ax=0$. Therefore,
it follows that
\begin{eqnarray*}
  \|r^k\|_2=\|(I-P_{\Lambda^k})Ax+(I-P_{\Lambda_k})e\|_2=\|(I-P_{\Lambda^k})e\|_2\leq\|e\|_2\leq\varepsilon.
\end{eqnarray*}

Second, assume that $T-\Lambda^k\neq \varnothing$, then it follows from the definition of the block-RIP and  \eqref{e14} that
\begin{eqnarray*}
  \|r^k\|_2&=&\|(I-P_{\Lambda^k})Ax+(I-P_{\Lambda^k})e\|_2\\
&\geq&\|(I-P_{\Lambda^k})Ax\|_2-\|(I-P_{\Lambda^k})e\|_2\\
&=&\|A_{T\cup \Lambda^k}\omega_{T\cup \Lambda^k}\|_2-\|(I-P_{\Lambda^k})e\|_2\\
&\geq&\sqrt{1-\delta_{|T\cup \Lambda^k|}}\|\omega_{T\cup \Lambda^k}\|_2-\|e\|_2\\
&\geq&\sqrt{1-\delta_{|T\cup \Lambda^k|}}\|x_{T-\Lambda^k}\|_2-\varepsilon\\
&\geq&\sqrt{1-\delta_{|T\cup \Lambda^k|}}\sqrt{|T-\Lambda^k|}\min_{i\in T}\|x[i]\|_2-\|e\|_2\\
&\geq&\sqrt{1-\delta_{NK+1}}\min_{i\in T}\|x[i]\|_2-\varepsilon\\
&>&\varepsilon.
\end{eqnarray*}
Therefore the  OMMP does not stop early.
The proof of  Theorem \ref{the4} is completed.
\end{proof}


\begin{thebibliography}{99}

\bibitem{h1} Cand\`{e}s E J, Romberg J K, Tao T. Robust uncertainty principles:
Exact signal reconstruction from highly incomplete frequency information.
IEEE Trans. Inf. Theory, 2006, 52(2): 489-509.
\bibitem{h2} Donoho D L. Compressed sensing. IEEE Trans. Inf. Theory,  2006, 52(4): 1289-1306.
\bibitem{2} Cand\`{e}s E J. The restricted isometry property and its implications for compressed sensing. Comptes Rendus Mathematique, 2008,  346(9-10): 589-592.
\bibitem{8}Cand\`{e}s E J,  Tao T T. Decoding by linear programming. IEEE Trans. Inf. Theory, 2005,  51(12): 4203-4215.
\bibitem{aa}Chen S, Billings S A, Luo W. Orthogonal least squarses methods and their application to non-linear system identification. Int. J. Contr., 1989, 50(5): 1873-1896.
\bibitem{11}Tropp J A, Gilbert A C. Signal recovery from random measurements via orthogonal matching pursuit. IEEE Trans. Inf. Theory, 2007, 53(12): 4655-4666.
\bibitem{cc2}Mo Q, Yi S. A Remark on the restricted isometry property in orthogonal matching pursuit. IEEE Trans. Inf. Theory, 2012, 58(6): 3654-3656.
\bibitem{cc4}Wang J, Shim, et al. On the recovery limit of sparse signals using orthogonal matching pursuit. IEEE Transactions on Signal Processing, 2012, 60(60): 4973-4976.
\bibitem{cc5}Wu R, Huang W, Chen D R. The exact support recovery of sparse signals with noise via orthogonal matching pursuit. IEEE Signal Processing Letters, 2013, 20(4): 403-406.
\bibitem{cc1}Dan W, Wang R H. Robustness of orthogonal matching pursuit under restricted isometry property. Science China Mathematics, 2014, 57(3): 627-634.
\bibitem{20} Wang J, Kwon S, Shim  B. Generalized orthogonal matching pursuit.  IEEE Trans. Signal Processing, 2012, 60(12): 6202-6216.
\bibitem{cc3}Dan W. Analysis of orthogonal multi-matching pursuit under restricted isometry property. Science China Mathematics, 2014, 57(10): 2179-2188.
\bibitem{12}Needell D,  Vershynin R. Signal recovery from incomplete and inaccurate measurements via regularized orthogonal matching pursuit. IEEE J. Sel. Topics Signal Processing, 2010, 4(2): 310-316.
\bibitem{121}Xu Z Q.  The performance of orthogonal multi-matching pursuit under RIP.  J. Comp. Math, 2015, 33: 495-516.
\bibitem{13}Donoho D L, Drori I, Tsaig Y, Starck J L.  Sparse solution of underdetermined linear equations by stagewise orthogonal matching pursuit. IEEE Trans. Inf. Theory, 2012, 58(2): 1094-1121.
\bibitem{22}Dai W, Milenkovic O. Subspace pursuit for compressive sensing signal reconstruction. IEEE Trans. Inf. Theory,  2009, 55(5): 2230-2249.
\bibitem{12'} Needell D, Troop J A.  CoSaMP: Itertive signal recovery from incomplete and inaccurate samples. Appl. Comput. Harmon. Anal., 2009, 26(3): 301-321.
\bibitem{g13} Eldar Y C, Mishali M.  Robust recovery of signals from a structured
union of subspaces. IEEE Trans. Inf. Theory, 2009, 55(11): 5302-5316.
\bibitem{c1}Eldar Y C, Mishali M. Block-sparsity and sampling over a union of subspaces. in pro. 16th Int. Conf. Digital Signal processing, 2009, 1-8.
\bibitem{g10} Parvaresh F, Vikalo H, Misra S, Hassibi B. Recovering sparse signals
using sparse measurement matrices incompressed DNA microarrays. IEEE
J. Sel. Top. Signal Process, 2008, 2(3): 275-285.
\bibitem{g11} Cotter S, Rao B. Sparse channel estimation via matching pursuit
with application to equalization. IEEE Trans. Commun, 2002, 50(3): 374-377.
\bibitem{g12} Mishali M, Eldar Y C. Blind multi-band signal reconstruction:
Compressed sensing for analog signals.  IEEE Trans. Signal
Processing, 2009, 57(3): 993-1009.
\bibitem{h3}Mishali M, Eldar Y C, Dounaevsky O, Shoshan E. Xampling:
Analog to digital at sub-Nyquist rates.  2009, arXiv 0912.2495.
\bibitem{c2} Mishali M, Eldar Y C. Reduce and boost: Recovering arbitrary sets of jointly sparse vectors. IEEE Trans.
Signal Processing, 2008, 56(10): 4692-4702.
\bibitem{c3} Cotter S F, Rao B D, Engan K, Kreutz-Delgado K.  Sparse solutions to linear inverse problems with multiple measurement vectors.
IEEE Trans. Signal Processing, 2005, 53(6): 2477-2488.
\bibitem{c4}Chen J, Huo  X. Theoretical results on sparse representations of
multiple-measurement vectors, IEEE Trans. Signal Processing, 2006,
54(12): 4634-4643.
\bibitem{c5}Tropp J A, Algorithms for simultaneous sparse approximation. Part I: Greedy pursuit, Signal Processing. (Special Issue on Sparse Approximations in Signal and Image Processing),
2006, 86: 572-588.
\bibitem{c6}Eldar Y C, Mishali M. Robust recovery of signals from a structured union of subspaces. IEEE Trans. Inf. Theory, 2008, 55(11): 5302-5316.
\bibitem{g14} Eldar Y C, Kuppinger P, H. B\"{o}lcskei. Block-sparse signals: uncertainty relations and efficient recovery. IEEE Trans. Signal Processing, 2010, 58(6): 3042-3054.
\bibitem{g13'}Stojnic M, Parvaresh F, Hassibi B. On the reconstruction of
block-sparse signals with an optimal number of measurements. IEEE
Trans. Signal Processing, 2010, 57(8): 3075-3085.

\bibitem{g13''}Lin J H, Li S. Block Sparse Recovery via Mixed $\ell_2/\ell_1$ Minimization. Acta Mathematica Sinica, 2013, 46(29):364-375.
\bibitem{cc6}Huang J, Zhang T, The benefit of group sparsity, Ann. Stat, 2010, 38(4): 1978¨C2004.
\bibitem{g15}Majumdar A, Ward R K, Compressed sensing of color images. Signal Processing, 2010, 90(12): 3122-3127.
\bibitem{g16}Wang Y, Wang J J, Xu Z B, On recovery of block-sparse signals via mixed $\ell_2/\ell_p(0<p\leq1)$ norm minimization, EURASIP J. Adv. Signal
Process, 2013,  76: 1-17.
\bibitem{g17}Wang Y, Wang J J, Xu Z B, Restricted $p$-isometry properties of nonconvex block-sparse compressed sensing,  Signal Processing, 2014, 104: 188-196.
\bibitem{g17'} Fu Y, Li H, Zhang Q, et al. Block-sparse recovery via redundant block OMP. Signal Processing, 2014, 97(7): 162-171.
\bibitem{g17''}Swirszcz G, Abe N, Lozano A C. Grouped orthogonal matching pursuit for variable selection and prediction. in Advances in
Neural Information Processing Systems, 2009, pp. 1150-1158.
\bibitem{g18'}Zhao Q, Wang J, Han Y, et al. Compressive sensing of block-sparse signals recovery based on sparsity adaptive regularized orthogonal matching pursuit algorithm. IEEE Fifth International Conference on Advanced Computational Intelligence. IEEE, 2012, pp. 1141-1144.
\bibitem{g18''}Huang B X, Zhou T. Recovery of block sparse signals by a block version of StOMP. Signal Processing, 2015, 109: 231-244.
\bibitem{g25}Xu Y, Qiu X H. Block-Sparse Signals Recovery using Orthogonal Multimatching. Journal of Signal Processing,  2014, 30(6): 706-711, .
\bibitem{g26}Chen W G, Ge H M, A sharp bound on RIC in generalized orthogonal matching pursuit. 2016, arXiv:1604.03306.
\bibitem{g27}Wen J M, Zhou Z C, Li D F, Tang X H. Improved sufficient conditions for sparse recovery with generalized orthogonal matching pursuit.  2016, arXiv:1603.01507.
\bibitem{16'}Mo Q, A sharp restricted isometry constant bound of orthogonal matching pursuit, 2015,
arXiv:1501.01708 v1[cs.IT].
\bibitem{g18} Wen J, Zhou Z, Liu Z, et al. Sharp sufficient conditions for stable recovery of block sparse signals by block orthogonal matching pursuit. 2016, arXiv:1605.02894.





\end{thebibliography}
\end{document}